\documentclass[10pt,journal,romanappendices]{IEEEtran}
 
\usepackage{amsmath,epsfig,amssymb,verbatim,amsopn,cite,subfigure,multirow}
\usepackage{balance, mathrsfs}
\usepackage{footnote}
\usepackage{algorithm,algorithmic}
\usepackage[algo2e]{algorithm2e} 
\usepackage[usenames,dvipsnames]{color}
\usepackage[all]{xy}  %%% used to make block diagram
\usepackage{url}
\usepackage{adjustbox}
\usepackage{enumitem}
\usepackage{cite,algorithm,algorithmic,amsmath,amssymb,amsthm,empheq,mhsetup}
\usepackage{subfigure,amsfonts,balance}
\usepackage{epstopdf}
\usepackage{setspace}
\usepackage[normalem]{ulem}
\usepackage{acronym}
\usepackage[most]{tcolorbox}
\usepackage[abs]{overpic}
\usepackage{rotating}
\usepackage{bm}

\newcommand\hl{\bgroup\markoverwith
    {\textcolor{yellow}{\rule[-.5ex]{.1pt}{2.5ex}}}\ULon}

\newtheorem{Proposition}{Proposition}

  {\proof}{\proofend}

\addcontentsline{toc}{section}{Appendix}

\newfont{\bb}{msbm10 scaled 1100}

\newcounter{mytempeqcounter}

\SetKwInput{KwInput}{Input}
\SetKwInput{KwOutput}{Output}

\definecolor{BgYellow}{HTML}{FFF59C}
\definecolor{FrameYellow}{HTML}{F7A600}

\newtcolorbox{StickyNote}[1][]{%
    enhanced,
    before skip=2mm,after skip=2mm, 
    width=\columnwidth, boxrule=0.2mm, % width of the sticky note
    colback=BgYellow, colframe=FrameYellow, % Colors
    attach boxed title to top left={xshift=0cm,yshift*=0mm-\tcboxedtitleheight},
    varwidth boxed title*=-3cm,
    boxed title style={frame code={%
        \path[left color=FrameYellow,right color=FrameYellow,
        middle color=FrameYellow]
        ([xshift=-0mm]frame.north west) -- ([xshift=0mm]frame.north east)
        [rounded corners=0mm]-- ([xshift=0mm,yshift=0mm]frame.north east)
        -- (frame.south east) -- (frame.south west)
        -- ([xshift=0mm,yshift=0mm]frame.north west)
        [sharp corners]-- cycle;
        },interior engine=empty,
    },
    sharp corners,rounded corners=southeast,arc is angular,arc=3mm,
    underlay={%
        \path[fill=BgYellow!80!black] ([yshift=3mm]interior.south east)--++(-0.4,-0.1)--++(0.1,-0.2);
        \path[draw=FrameYellow,shorten <=-0.05mm,shorten >=-0.05mm,color=FrameYellow] ([yshift=3mm]interior.south east)--++(-0.4,-0.1)--++(0.1,-0.2);
        },
    drop fuzzy shadow, % Shadow
    fonttitle=\bfseries, 
    title={#1}
}

\begin{document}

\title{Pinching-Antenna System-Assisted Localization: \\ A Stochastic Geometry Perspective}

\author{Jiajun~He,~\IEEEmembership{Member, IEEE}, Xidong Mu,~\IEEEmembership{Member, IEEE,}~Hien Quoc Ngo,~\IEEEmembership{Fellow, IEEE},\\~and~Michail Matthaiou,~\IEEEmembership{Fellow, IEEE}}
        
% The paper headers
\markboth{}%
{Shell \MakeLowercase{\textit{et al.}}: Bare Demo of IEEEtran.cls for IEEE Journals}

\maketitle

\begin{abstract}
This paper proposes a novel localization framework underpinned by a pinching-antenna (PA) system, in which the target location is estimated using received signal strength (RSS) measurements obtained from downlink signals transmitted by the PAs. To develop a comprehensive analytical framework, we employ stochastic geometry to model the spatial distribution of the PAs, enabling tractable and insightful network-level performance analysis. Closed-form expressions for target localizability and the Cramér–Rao lower bound (CRLB) distribution are analytically derived, enabling the evaluation of the fundamental limits of PA-assisted localization systems without extensive simulations. Furthermore, the proposed framework provides practical guidance for selecting the optimal waveguide number to maximize localization performance. Numerical results also highlight the superiority of the PA-assisted approach over conventional fixed-antenna systems in terms of the CRLB.
\end{abstract}

\begin{IEEEkeywords}
Cram$\acute{\text{e}}$r-Rao lower bound, localization, pinching antenna, received signal strength.
\end{IEEEkeywords}

\IEEEpeerreviewmaketitle

\section{Introduction}

\IEEEPARstart{P}{inching}-antenna (PA) systems have garnered significant attention in recent years due to their reconfigurability and practical deployment potential \cite{ding2025PA}. Compared to traditional multiple-antenna systems, a PA system comprises multiple discrete dielectric particles, referred to as PAs, that are pinched on dielectric waveguides to facilitate signal transmission. Equally importantly, DOCOMO demonstrated in \cite{DOCOMO} that a PA system can be implemented in a cost-effective manner. Motivated by these promising characteristics, PA systems have been extensively studied in the literature to enhance communication functionality. The fundamental performance limits of PA systems under various network configurations were thoroughly examined in \cite{DingLoS}, demonstrating that a PA system can achieve significant performance gains over conventional antenna systems. In addition, a distinguishing feature of PA systems is the reconfigurability of the PAs. By adjusting their positions along the waveguide, the system can simultaneously serve multiple mobile users and enhance communication performance by minimizing the distance between each user and its associated PA.

Due to the advantages of PA systems, recent studies have investigated their potential to improve localization performance. For instance, Ding \cite{ding2025pinchingantennaassistedisaccrlb} analyzed the fundamental limits of a PA-assisted time-of-arrival (ToA) localization system in terms of the Cramér–Rao lower bound (CRLB), demonstrating that the use of PAs can significantly reduce localization error. Khalili $et$ $al.$ \cite{khalili2025pinchingantennaenabledisacsystems} investigated the large-scale deployment of waveguides to improve target diversity and sensing performance. Qin $et$ $al.$ \cite{Fu2025} investigated how a PA system can be used to enhance the performance of integrated sensing and communication (ISAC) systems in terms of communication and sensing rates, whereas the inner and outer bounds of the achievable communication–sensing rate region were derived in \cite{Ouyang2025CR} to gain insights into the information-theoretic limits of PA-assisted ISAC systems. Although PA systems have demonstrated strong potential in enhancing sensing and localization performance, a significant challenge cannot be overlooked when applying them to target localization. The sparse deployment of PAs along dielectric waveguides may limit their ability to capture a wide range of signals transmitted by mobile users. The signal reception on a single waveguide equipped with multiple PAs can lead to significant inter-antenna coupling, which lowers the localization accuracy. To address this issue, prior works, such as \cite{zhou2025channelestimationmmwavepinchingantenna} and \cite{wang2025wirelesssensingpinchingantennasystems}, have proposed the usage of additional antenna arrays or leaky coaxial cables (LCCs) to capture echo signals reflected by the target for channel estimation and target localization. 

From a practical implementation perspective, this work considers received signal strength (RSS)-based localization within a PA system, where multiple PAs transmit downlink signals to the target for localization.\footnote{Different from ToA- and angle-of-arrival (AoA)-based localization, which typically require wideband signals or antenna arrays to obtain range and angle measurements, RSS-based localization offers a lower-complexity alternative by eliminating the need for such hardware and reducing the burden of acquiring location-dependent features \cite{hcso-RSS}.} In addition, to the best of our knowledge, existing studies on PA system-assisted localization primarily focus on system-level or link-level designs and fail to provide insights into how different system configurations affect the overall localization performance. Motivated by this limitation, we aim to develop a stochastic geometric framework that studies the impact of the spatial distribution of PAs on the localization performance, thereby enabling network-level performance analysis and system optimization. The main contributions of this paper are listed as follows:

\begin{enumerate}
    \item \textit{Stochastic Modeling of the PA System}: A unified evaluation framework is developed using stochastic geometry, providing valuable insights into the design of PA-assisted localization systems from a network-level perspective. To our knowledge, this is the first work to leverage the Poisson line process (PLP) for analyzing the localization performance of PA systems.

    \item \textit{Localizability and CRLB}: To assess whether the mobile user (target) can be localized with a sufficient signal-to-interference-plus-noise ratio (SINR), the target localizability is investigated. Tractable expressions for the CRLB and its distribution are then derived, offering valuable insights into how network configurations and channel characteristics influence the localization performance.

    \item \textit{Guideline for System Design}: A comprehensive evaluation is conducted to gain insights into the fundamental limits of PA-assisted localization. Numerical results show that the PA-assisted system achieves significantly improved localization accuracy compared to that of using uniform linear arrays (ULAs). In addition, the proposed stochastic geometric framework offers practical guidelines for optimizing the number of waveguides and configuring network parameters to meet the desired localization requirements.
\end{enumerate}

%\textit{Notation}: Bold lower-case letters represent column vectors, while capital letters are matrices. The superscripts $(\cdot)^T$ and $(\cdot)^{-1}$, and $\text{tr}(\cdot)$ represent the transpose, inverse, and trace, respectively; $\|\cdot\|$ and $\mathbb{E}\{\cdot\}$ are the Euclidean norm and expectation operator. The estimate of the variable $x$ is $\hat{x}$. 

\vspace{-0.4cm}

\section{System Model}

We consider a localization problem (see \ref{fig-system model PLP}) in a two-dimensional (2D) space, where multiple PAs on waveguides transmit downlink signals to a target user for localization.   

\vspace{-0.4cm}

\subsection{Network Model}

According to Slivnyak’s theorem  \cite{stoyan2013stochastic}, we assume that the target $\mathbf{p}_t = [x_t \ y_t]^T$ is located at the origin \textit{O}, while the waveguides are randomly distributed within a 2D disk with center point at the origin and radius $R_a$ according to a PLP with line density $\lambda_l$. Specifically, the $k$-th waveguide is modeled by an undirected line $L_{k}$, which is: $L_k(\rho_k, \theta_k) = \left \{ (x,y) \in \mathbb{R}^2 : x\cos\theta_k + y\sin\theta_k = \rho_k  \right \}$, where $\rho_k$ denotes the perpendicular distance from the origin \textit{O} to the $k$-th waveguide, while $\theta_k$ is its inclination angle. During localization, the $m$-th PA on waveguide $k$ is activated from the predefined locations ($\mathbf{p}_{k,m} = [x_{k,m} \ y_{k,m}]^T$, where $k = 1, 2, \ldots$ and $m = 0, 1, \ldots$), while the predefined locations on the waveguide are modeled by a one-dimensional (1-D) homogeneous Poisson point process (HPPP) with node density $\lambda_s$. By indexing the waveguides in ascending order of their perpendicular distance $\rho_k$ from the target, the probability density function (PDF) and cumulative distribution function (CDF) of the distance $\rho_k$ to the $k$-th nearest waveguide are \cite{sun2020plp}
\begin{equation}
    \begin{split}
        f_{\rho_k}(x) &= \frac{e^{-2\pi \lambda_l x}(2\pi \lambda_l x)^k}{x(k-1)!}, \\
        F_{\rho_k}(x) &= 1 - e^{-2\pi \lambda_l x} \sum_{n=0}^{k-1} \frac{(2\pi \lambda_l x)^n}{n!}. \\
    \end{split}
\label{distribution - rho}
\end{equation}
Let $d_{k,1}= \|\mathbf{p}_{k,1} - \mathbf{p}_t \|$, where $\|\cdot\|$ denotes the Euclidean norm, be the distance between the target and the nearest PA ($m=1$) on the $k$-th waveguide. Then, conditioned on $\rho_k$, the PDF and CDF of $d_{k,1}$ are, respectively, given by
\begin{equation}
    \begin{split}
        f_{d_{k,1}|\rho_k}(r) &= \frac{2\lambda_s r}{\sqrt{r^2 - \rho_k}} e^{-2\lambda_s \sqrt{r^2 - \rho_k}}, \\
        F_{d_{k,1}|\rho_k}(r) &= 1 - e^{-2\lambda_s \sqrt{r^2 - \rho_k}}.
    \end{split}
\label{d-CDF}
\end{equation}

\subsection{Signal Model}

\begin{figure}[t]
\centerline{\includegraphics[width=0.62\columnwidth]{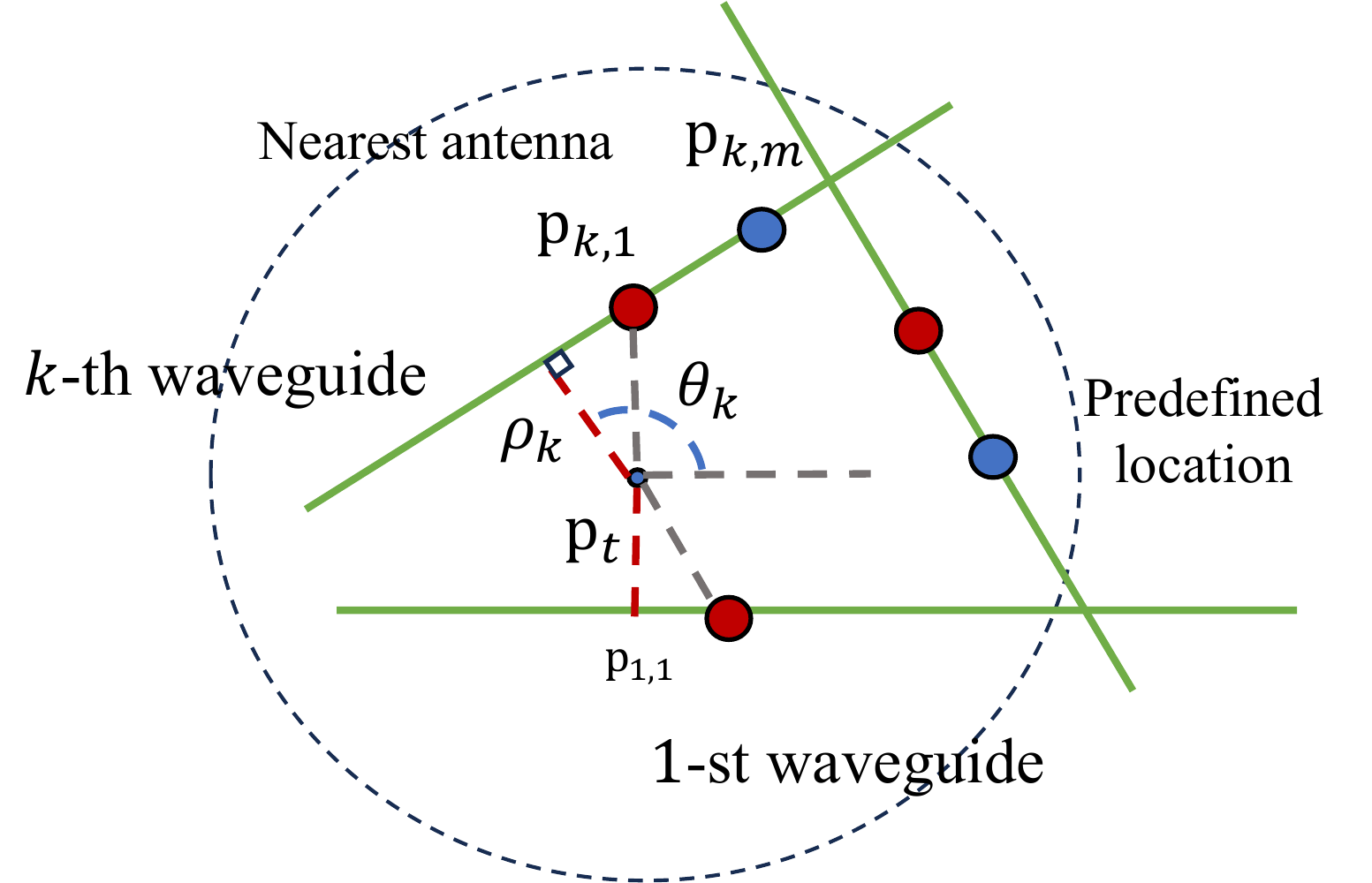}}
\caption{Illustration of the PLP-based PA-assisted localization system.}
\label{fig-system model PLP}
\end{figure}

Each waveguide activates its nearest PA to the target user, denoted as $\mathbf{p}_{k,1}$, during the localization process. Although all PAs serve the same target user for localization, the RSS sample has to be extracted from a specific PA to obtain distance information \cite{hcso-RSS, so-loc}. Consequently, each PA on a different waveguide transmits different symbols, which introduces interference and compromises the localization performance. For example, in cellular networks, the base station (BS) transmits signals containing various types of information, such as the cell-ID, that enable mobile users to extract location information from the corresponding BS \cite{so-loc}. Thus, the  signal received at the target user is
\begin{equation}  \label{eq-received signal}
    y = \sum_{\substack{k=1}} \sqrt{P_t} h_{k,1} s_{k} + w_n,
\end{equation}
where $P_t$ is the transmit power at each PA, while $s_{k}$, where $\mathbb{E}\{|s_{k}|^2\}=1$, is the symbol transmitted by the $k$-th waveguide, and $\mathbb{E}\{\cdot\}$ is the expectation operation; $h_{k,1}$ is the channel gain between the target and its nearest PA on the $k$-th waveguide, which is given by \cite{DingLoS}
\begin{equation}
    h_{k,1} = \frac{\sqrt{\eta} e^{-2\pi j \left (\frac{1}{\lambda}\|\mathbf{p}_t - \mathbf{p}_{k,1}\|+\frac{1}{\lambda_g}\|\mathbf{p}_{k,0} - \mathbf{p}_{k,1}\| \right )}}{\|\mathbf{p}_t - \mathbf{p}_{k,1}\|^{\alpha/2}},
\label{eq-channel}
\end{equation}
where $\eta = c^2/ 16\pi^2 f_c^2$, $c$ is the speed of light, $f_c$ is the carrier frequency, while $\lambda$ and $\lambda_g$ are the carrier and waveguide wavelength, respectively. Also, $w_n$ is the additive white Gaussian noise (AWGN) with variance $\sigma^2$, while $\alpha$ is the path-loss exponent. From \eqref{eq-received signal} and \eqref{eq-channel}, the SINR at the target associated with its nearest PA on the $k$-th waveguide is
\begin{equation}
\begin{split}
    \!\!\! \text{SINR}_{k}^{\text{PA}} &= \frac{P_t|h_{k,1} |^2}{\sum_{\substack{i=1, i\neq k}}^{} P_t|h_{i,1} |^2}
    \nonumber\\
    &=\frac{\eta P_t d_{k,1}^{-\alpha}}{\underbrace{\sum_{\substack{i=1, i\neq k}}^{} \eta P_t d_{i,1}^{-\alpha}}_{\text{interference from remaining PAs}} + \sigma^2} = \frac{d_{k,1}^{-\alpha}}{I + \sigma_n^2}, \\
\end{split}
\label{eq-SINR update}
\end{equation}
where $I = \sum_{\substack{i=1, i\neq k}}^{} d_{i,1}^{-\alpha}$, while $\sigma_n^2 = \frac{\sigma^2}{P_t \eta}$ is the normalized noise power. Since a line-of-sight (LoS) path between the target and the PA can be achieved by adjusting the PA location, the path-loss exponent typically ranges from $2$ to $2.8$ \cite{so-loc}.

\vspace{-0.5cm}

\subsection{RSS Measurements}

Considering a practical implementation of the PA-assisted localization system, the RSS-based scheme is considered. By assuming that the nearest predefined location on each waveguide has been determined, the RSS measured by the target user associated with the activated PA on the $k$-th waveguide is \cite{so-loc, hcso-RSS}
\begin{equation}
    r_{\text{RSS},k,1} = \ln P_{r,k,1} - \ln \eta - \ln P_t =  -\alpha \ln d_{k,1} + n_{p},
\label{eq-RSS model}
\end{equation}
where $P_{r,k,1}$ is the received signal power at the target user transmitted by the activated PA on the waveguide $k$, while $n_{p}$ is Gaussian distributed with variance $\sigma_{p}^2 = \frac{\ln 10}{10\alpha\mathbb{E}\{\text{SINR}_{k}^{\text{PA}}\}}$. Note that the value of $\alpha$ is normally obtained through field testing and calibration campaigns \cite{hcso-RSS}. By collecting sufficient RSS measurements, the location estimate can be obtained using maximum likelihood estimation (MLE) as described in \cite{hcso-RSS}.

\vspace{-0.4cm}

\section{Performance Analysis}

\subsection{Target Localizability}

To determine the target localizability, we evaluate the probability of detecting \textit{at least} $K$ waveguides during localization. Since the successful reception from the activated PA on the $K$-th nearest waveguide implies that all nearer waveguides are within the detection range, we investigate the probability that the SINR between the target and the selected PA on the $K$-th nearest waveguide exceeds a given threshold $\tau$, given by
\begin{equation}
\begin{split}
    \!\!\! P(\text{SINR}_{K}^{\text{PA}} > \tau) = P \left ( \frac{\eta P_t d_{K,1}^{-\alpha}}{{\sum_{\substack{i=1, i\neq K}}^{} \eta P_t d_{i,1}^{-\alpha}} + \sigma^2} > \tau \right ). \\
\end{split}
\label{eq-SINR multi}
\end{equation}
To derive a closed-form expression of \eqref{eq-SINR multi}, the dominant interference analysis is adopted \cite[Assumption 1]{he-rss}. Specifically, all interferers, except for the PA located on the nearest waveguide to the target, are approximated by their mean values. Then, we have the following proposition.

\begin{Proposition}
\label{prop-PA SINR}
    Considering a PA system with multiple waveguides, where the nearest predefined location to the target is activated to transmit downlink signals for localization, the probability of detecting at least $K$ waveguides is 
    \begin{equation}
    \begin{split}
        &\quad P({{\rm{SINR}}}_{K}^{\text{PA}} > \tau | d_{1,1}, d_{K,1}, \rho_{K}) \\ 
        &\approx \int_0^{R_a} \int_{0}^{d_{K}} F_{d_{K,1}|\rho_{K}} \left ( \left [\tau \left ( \mathbb{E}\{I\} + \sigma_n^2 \right ) \right ]^{-\frac{1}{\alpha}} \right ) \\ &\quad\quad \times f_{d_{1,1}}(r)f_{d_{K,1}|\rho_{K}}(x) f_{\rho_k}(y) dr dx dy  , \\
    \end{split}
    \label{eq-SINR multi close-form}
    \end{equation}
    where $f_{d_{1,1}}(\cdot)$, $f_{d_{K,1}|\rho_{K}}(\cdot)$, $F_{d_{K,1}|\rho_{K}}(\cdot)$, and $f_{\rho_k}(\cdot)$ are provided in \eqref{distribution - rho} and \eqref{d-CDF}. Conditioned on the distances $d_{1,1}$, $d_{K,1}$, $\rho_{K}$, the interference $I$ can be approximated by
    \begin{equation}
    \begin{split}
        &\quad \mathbb{E} \{I_{K} | d_{1,1}, d_{K,1}, \rho_1, \rho_{K} \} \\
        &\approx \begin{cases}
        &d_{1,1}^{-\alpha} + \frac{2(K-2)}{2-\alpha} \frac{\rho_{K}^{2-\alpha} - \rho_{1}^{2-\alpha}}{\rho_{K}^{2} - \rho_{1}^{2}} + \frac{2\pi\lambda_{l}}{\alpha-2} \rho_{K}^{2-\alpha}, \alpha>2,\\
        &d_{1,1}^{-2} + \frac{2(K-2)}{\rho_{K}^{2} - \rho_{1}^{2}} \ln \frac{\rho_k}{\rho_1} + 2\pi\lambda_{l}\ln\frac{R_a}{\rho_k},  \quad \quad \alpha=2.
        \end{cases}
    \end{split} 
    \end{equation} 
\end{Proposition}

\begin{proof}
    See Appendix \ref{Prop-1}.
\end{proof}

\vspace{-0.2cm}

For RSS-based localization, at least three waveguides are required to localize the target position without ambiguity \cite{hcso-RSS}. \textbf{Proposition \ref{prop-PA SINR}} characterizes the probability of successfully conducting the localization procedure and illustrates how this metric is influenced by various network parameters. Based on \textbf{Proposition \ref{prop-PA SINR}}, the expected SINR is computed by: $\mathbb{E}\{ \text{SINR}_{k}^{\text{PA}}  \} = \int_{0}^{\infty} P( \text{SINR}_{k}^{\text{PA}} > \tau) d\tau$. Based on the expected SINR, the RSS disturbance $\sigma_p^2$ is obtained.

\vspace{-0.4cm}

\subsection{CRLB Analysis}

Based on the RSS signal model given by \eqref{eq-RSS model}, the Fisher information matrix (FIM) is given by
\begin{equation}
\begin{split}
&\text{FIM}_{\rm{RSS}}(\mathbf{p}_t) = \frac{1}{\sigma_{\text{RSS}}^2} \begin{bmatrix}
A & C \\
C & B \\
\end{bmatrix}, \\
\end{split}
\label{eq-39-1}
\end{equation}
\begin{equation}
\begin{split}
A &= \sum_{k=1}^{K}\frac{(x_{k,1}-x_{t})^{2}}{d_{k,1}^{4}}, \
B = \sum_{k=1}^{K} \frac{(y_{k,1}-y_{t})^{2}}{d_{k,1}^{4}}, \\
C &= \left ( \sum_{k=1}^{K} \frac{(x_{k,1}-x_{t})(y_{k,1}-y_{t})}{d_{k,1}^{4}} \right )^2,
\end{split}
\label{eq-41}
\end{equation}
where $\sigma_{\text{RSS}} = \frac{\sigma_{p}}{\alpha}$, and $x_t = y_t = 0$. Note that the following analysis remains valid even when the target is not located at the origin. The CRLB for RSS-based localization is then:  
\begin{equation}
    \begin{split}
       {\rm{CRLB}}_{\text{RSS}}(\mathbf{p}_t) =
       {\rm{tr}} \left ({{\text{FIM}}}_{{\rm{RSS}}}^{-1}(\mathbf{p}_t)\right ) = \sigma_{\text{RSS}}^2 \frac{A+B}{AB - C^2}.
    \end{split}
\label{eq-RSS CRLB exact}
\end{equation}
where $\text{tr}(\cdot)$ is the trace operation. In addition to deriving the CRLB using only the RSS samples detected from the nearest PA on each waveguide, the PA can be activated at different timestamps at all predefined locations on the waveguide for localization. In this case, the variables $A$, $B$, and $C$, can then be updated as follows:
\begin{equation}
\begin{split}
\!\!\!\!\bar{A} &= \sum_{k=1}^{K}\sum_{m=0}^{M-1}\frac{(x_{k,m}-x_{t})^{2}}{d_{k,m}^{4}}, \
\bar{B} = \sum_{k=1}^{K}\sum_{m=0}^{M-1} \frac{(y_{k,m}-y_{t})^{2}}{d_{k,m}^{4}}, \\
\bar{C} &= \left ( \sum_{k=1}^{K} \sum_{m=0}^{M-1} \frac{(x_{k,m}-x_{t})(y_{k,m}-y_{t})}{d_{k,m}^{4}} \right )^2,
\end{split}
\label{eq-update 41}
\end{equation}
where $M$ is the number of predefined locations on each waveguide. By substituting \eqref{eq-update 41} into \eqref{eq-RSS CRLB exact}, the RSS-based CRLB, considering all predefined locations for PA activation, can be obtained. However, \eqref{eq-RSS CRLB exact} involves the random variables $x_{k,m}$, $y_{k,m}$, and $d_{k,m}$, which makes it challenging to derive the distribution of the CRLB. Motivated by this challenge, we employ the concept of mutual information to derive a tractable approximation of the CRLB, as presented in the following.

\begin{Proposition}
\label{prop-PA RSS CRLB}
By assuming that $K$ waveguides are selected for localization, while each waveguide has $M$ predefined PA positions. Given the RSS disturbance $\sigma_{p}^2$, the RSS CRLB can be approximated by
\begin{equation}
{\rm{CRLB}}_{\text{RSS}}(\mathbf{p}_t) \approx \sigma_{\text{RSS}}^2 \frac{4}{M(K-1)}d_{*}^{2},
\label{eq-CRLB multi}
\end{equation}
where $d_{*} \in \{ d_{k,m} \}_{k=1, m=0}^{K,M-1}$ is the selected distance between the target user and the PA activated at the $m$-th predefined location on the $k$-th waveguide. Considering a practical scenario in which only a single PA on the waveguides is activated for localization, we can set $M=1$.
\end{Proposition}

\begin{proof}
    See Appendix \ref{Prop-RSS CRLB}.
\end{proof}

\vspace{-0.4cm}

Compared to the lengthy and complex expression in \eqref{eq-RSS CRLB exact}, the approximate CRLB presented in \textbf{Proposition \ref{prop-PA RSS CRLB}} provides a concise and analytically tractable expression that effectively captures the localization performance. This simplified form serves as a valuable tool for performance evaluation without relying on computationally intensive simulations. It shows that \eqref{eq-CRLB multi} is a function of distance $d_{*}$, while the CDF of $d_{*}$ is given in \eqref{d-CDF}. The conditional RSS-based CRLB distribution is
\begin{equation}
\label{eq-crlb distribution}
\begin{split}
    &\quad P\left ( {\rm{CRLB}}_{\text{RSS}}(\mathbf{p}_t) \leq s \right | M, K, \sigma_{p}^2) \\
    &= P ( d_{*} \leq {\sqrt{sM(K-1)}}/{2\sigma_{p}} | M, K, \sigma_{p}^2 ). \\
\end{split}
\end{equation}
By substituting the CDF of $d_{*}$ into \eqref{eq-crlb distribution}, the CRLB distribution for the PA system is obtained. Compared to conventional analytical methods that assume a fixed network geometry, the stochastic geometric framework accommodates random network deployments. By leveraging the inherent randomness of the considered PA system, the CRLB distribution provides insights into how the number of waveguides and channel parameters affect the overall localization performance across all possible geometric configurations.

Although this study considers RSS for localization, it was suggested in \cite{ding2025pinchingantennaassistedisaccrlb} that ToA measurements can be utilized for target localization. However, the time synchronization among the deployed waveguides cannot be overlooked during the localization process. To mitigate this impact, time-difference-of-arrival (TDoA) localization can be applied. In our previous work \cite{HETDoA}, the approximate CRLB for the TDoA-based localization is: $ {\rm{CRLB}}_{\text{TDoA}}(\mathbf{p}_t) \approx \frac{\sigma_{\tau}^2 2(K-1)}{\sum_{k=2}^{K}(1+\cos^2\theta_k-2\cos\theta_{k})}$, where $\theta_k$ is the angle from the nearest PA to the target on the $k$-th waveguide, while the range variance $\sigma_{\tau}^2$ is computed by following \cite{so-loc}. The performance of the RSS-based method will be compared with that of using TDoA. For angle-based localization, estimating the AoA typically requires the use of antenna arrays \cite{zhou2025channelestimationmmwavepinchingantenna}, while a single PA or a target equipped with only a single antenna cannot extract angle information. 

\vspace{-0.3cm}

\section{Numerical Results}

\begin{figure}[t]
\centerline{\includegraphics[width=0.65\columnwidth]{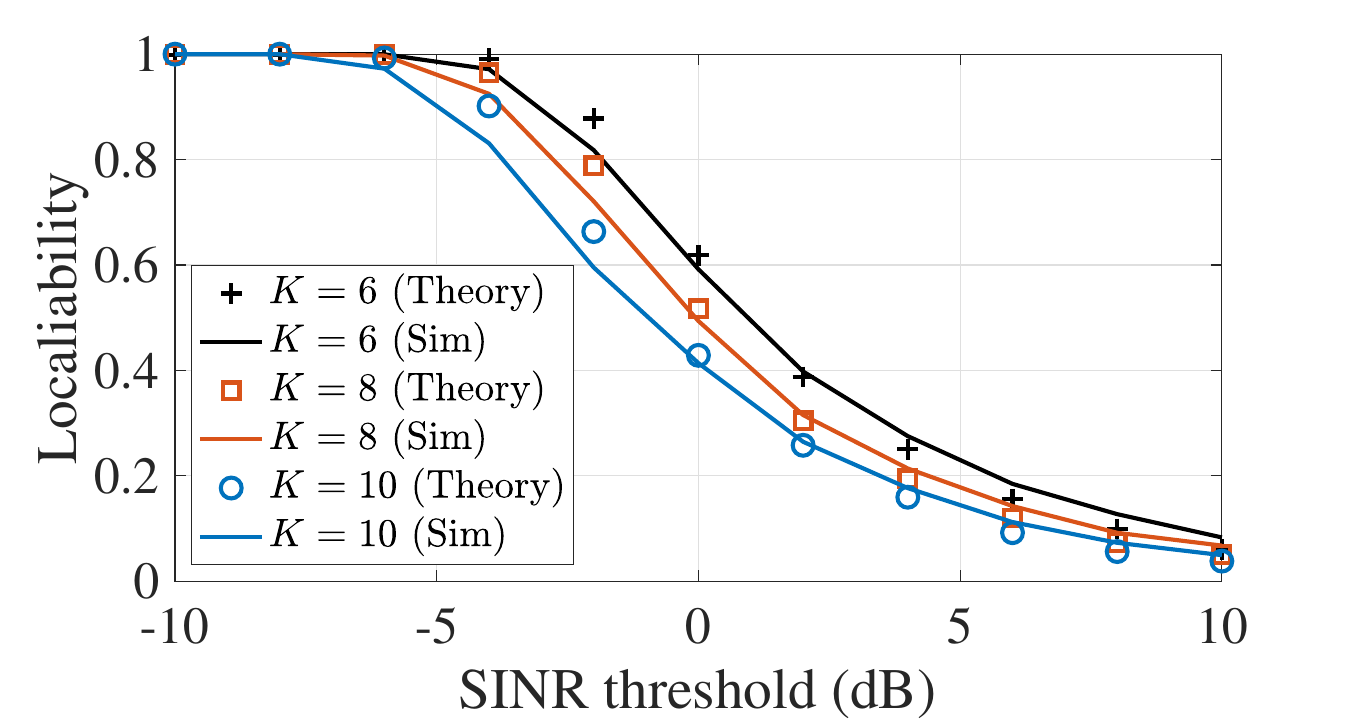}}
\centering
\caption{Accuracy of \textbf{Proposition \ref{prop-PA SINR}} versus different SINR thresholds.}
\label{fig-coverage}
\end{figure}

\begin{figure}[t]
\centerline{\includegraphics[width=0.65\columnwidth]{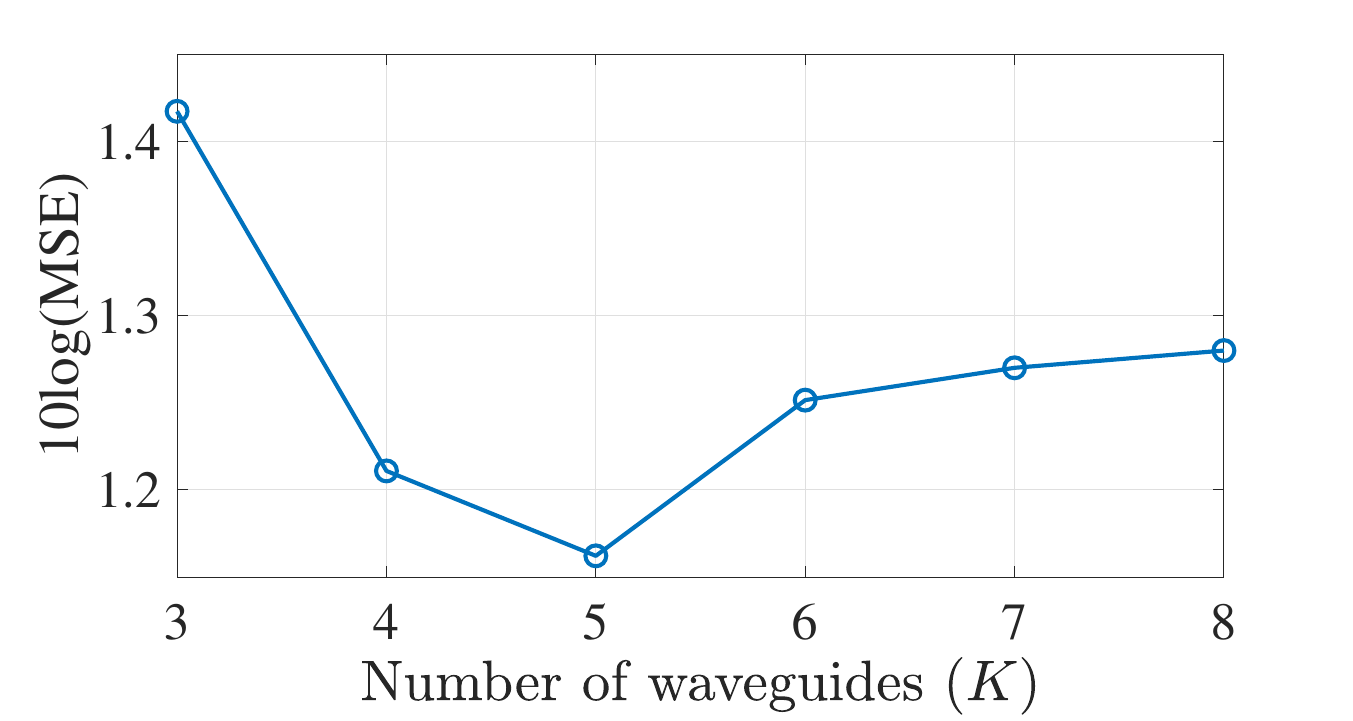}}
\centering
\caption{The localization performance versus different numbers of waveguides.}
\label{fig-CRLB SINR tradeoff}
\end{figure}

\begin{figure}[t]
\centerline{\includegraphics[width=0.65\columnwidth]{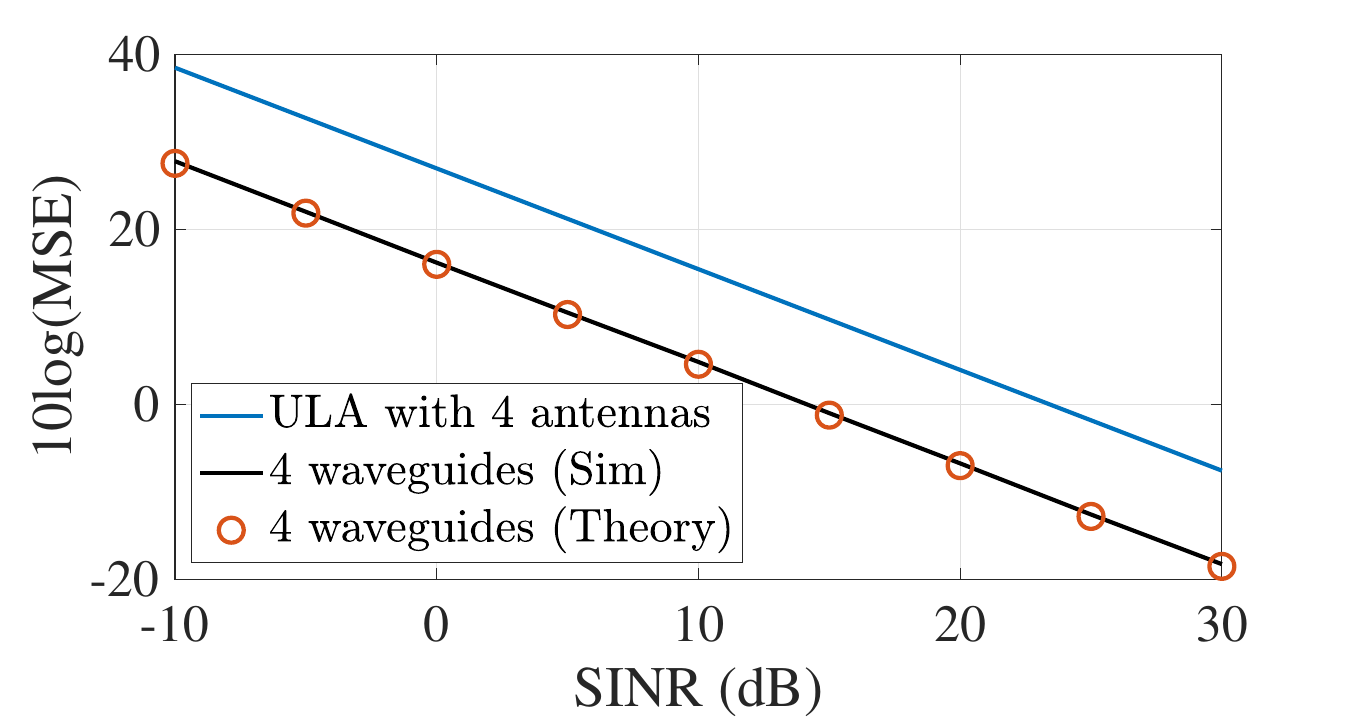}}
\centering
\caption{Accuracy of \textbf{Proposition \ref{prop-PA RSS CRLB}} versus different SINR values.}
\label{fig-CRLB accuracy}
\end{figure}

\begin{figure}[t]
\centerline{\includegraphics[width=0.645\columnwidth]{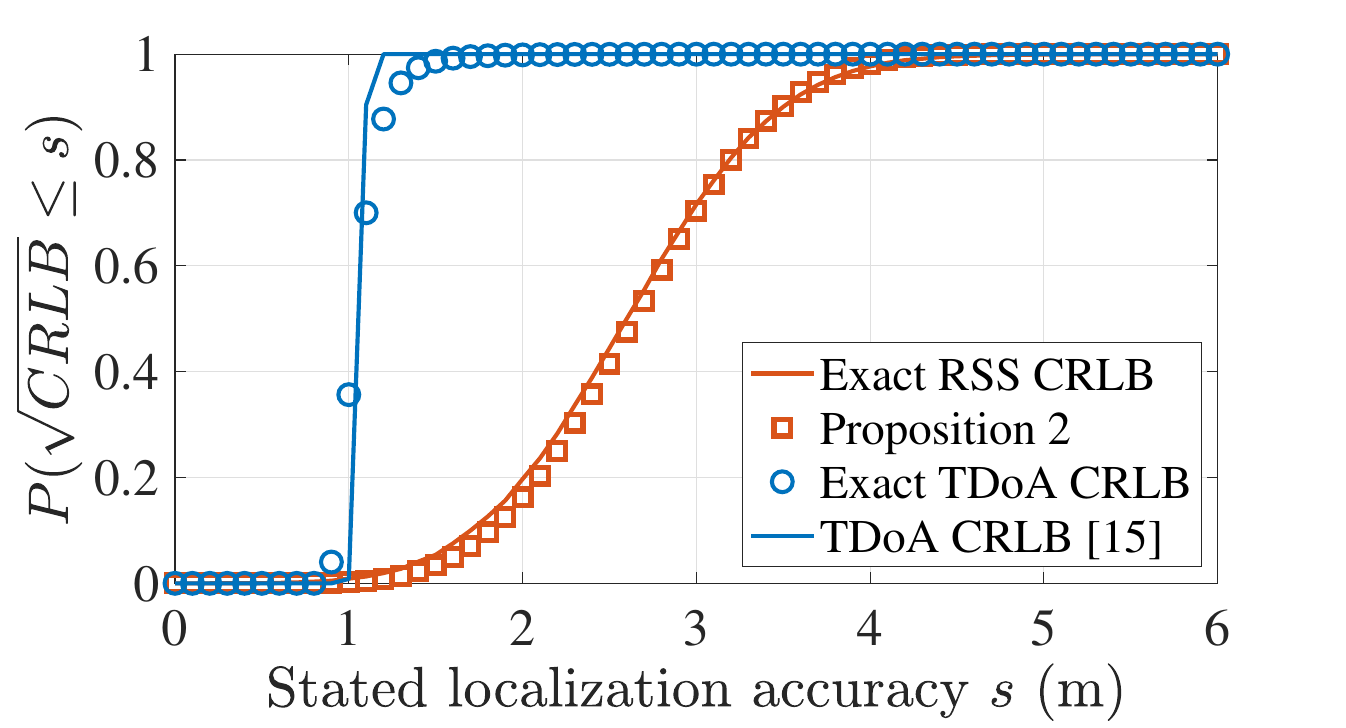}}
\centering
\caption{Accuracy of CRLB distribution versus localization requirements.}
\label{fig-CRLB distribution}
\end{figure}

\subsection{Simulation Setups}

For the simulation setup, waveguides are modeled using a PLP with line density $\lambda_l = 0.1/\pi~\mathrm{m}^{-1}$ in the 2-D disk with center point at the origin and radius $R_a = 30$ m, while the predefined locations on each waveguide are randomly distributed according to a 1-D HPPP with density $\lambda_s = 0.1~\mathrm{m}^{-1}$. Each PA transmits signals with a constant power of $P_t = 1$ W, and the path-loss exponent is set to $\alpha = 2.1$ \cite{wang2025wirelesssensingpinchingantennasystems}. All simulation results are averaged over $10^5$ independent runs.

\vspace{-0.4cm}

\subsection{Tradeoff between Localizability and CRLB}

Figure \ref{fig-coverage} illustrates the localizability performance of PA systems versus different SINR values. It is observed that localizability decreases as the required SINR increases. Although a higher SINR value guarantees better RSS quality, it is crucial to optimize network design to mitigate the impact of interference on localization performance. As the number of available waveguides increases, the overall localizability deteriorates. It is understandable that when different PAs transmit distinct pilots to the target for extracting RSS samples, an increase in the number of PAs results in greater interference.

Figure \ref{fig-CRLB SINR tradeoff} provides insights into how the number of available waveguides $K$ affects the overall localization performance. It is evident that increasing the number of waveguides does not always guarantee improved localization accuracy. When the number of available waveguides is limited (e.g., $K = 3$ to $5$), the localization performance improves with $K$ since the benefit of additional measurements outweighs the effect of interference. Notably, when $K = 5$, the RSS-based scheme achieves the lowest mean-square error (MSE), indicating the best localization performance. However, further increasing the number of waveguides results in worse localization performance, as interference becomes the dominant factor. Therefore, it is essential to determine the optimal number of waveguides for localization to balance the trade-off between interference and localization accuracy. 

\vspace{-0.5cm}

\subsection{Overall Localization Performance}

Figure \ref{fig-CRLB accuracy} examines the performance of the RSS-based localization across different SINR values. The theoretical results generated by \textbf{Proposition \ref{prop-PA RSS CRLB}} are consistent with the exact RSS-based CRLB obtained from \eqref{eq-RSS CRLB exact}. Furthermore, we compare PA system-assisted localization with conventional ULA-based localization. By adjusting the PA locations, the LoS distance between the PA and the target user is shortened, resulting in the acquisition of higher-quality RSS samples. The PA system also provides greater spatial diversity, whereas the ULA offers location-related measurement observations from only a single standpoint. Thus, the PA system systematically achieves better localization performance than the ULA. 

Figure \ref{fig-CRLB distribution} shows that the theoretical results generated by \eqref{eq-crlb distribution} closely match the simulation results, verifying the theoretical development of the CRLB distribution. Compared to the exact CRLB \eqref{eq-RSS CRLB exact}, the CRLB distribution provides insights into the overall localization performance. For instance, Fig. \ref{fig-CRLB distribution} shows that the RSS-based scheme achieves $3$-meter localization accuracy with a probability of $40\%$. We also compare the RSS-based scheme with the TDoA-based method \cite{HETDoA}. Since RSS measurements are less stable in signal propagation environments, while high-resolution TDoAs can be leveraged for localization, the TDoA-based scheme in PA systems outperforms that of using RSS. However, in practice, obtaining accurate time measurements in a PA system without the assistance of additional receivers, such as a ULA or LCC, remains an open challenge. Thus, the RSS-based localization scheme remains a favorable solution for PA systems.

\vspace{-0.4cm}

\section{Conclusion}

This paper presented a unified analytical framework for PA-assisted localization using stochastic geometry, modeling waveguides via a PLP. Our framework enables network-level performance analysis without intensive simulations and provides insights into selecting the optimal number of waveguides to achieve the optimal localization performance in terms of target localizability and CRLB. Numerical results showed that the PA-assisted system outperforms fixed-antenna systems, underscoring the benefits of using PAs for accurate positioning.

\vspace{-0.5cm}

\appendices

\section{}
\label{Prop-1}

When $\alpha > 2$, given distances $d_{1,1}$, $d_{K,1}$, and $\rho_k$, the probability of detecting $K$ waveguides is computed by
\begin{equation}
\begin{split}
    &\quad P(\text{SINR}_{K}^{\text{PA}} > \tau | d_{1,1}, d_{K,1}, \rho_{K}) \\
    &\approx F_{d_{K,1}|\rho_K} \left ( \left [\tau \left ( \mathbb{E}\{I\} + \sigma_n^2 \right ) \right ]^{-\frac{1}{\alpha}} \right ), \\
\end{split}
\label{eq-SINR multi 1}
\end{equation}
while the interference from the remaining $K-2$ PAs (excluding the PA on the nearest waveguide to the target) is
\begin{equation}
\begin{split}
    &\mathbb{E}\bigg \{ \sum_{\substack{i=2}}^{K-1}  d_{i,1}^{-\alpha} \ \bigg | \ d_{1,1}, d_{K,1}, \rho_1, \rho_{K}  \bigg \} \\
    =& \sum_{\substack{i=2}}^{K-1} \int_{d_{1,1}}^{d_{i+1,1}} r^{-\alpha} f_{d_{i,1}|\rho_{i}}(r) f_{\rho_{i}}(x) drdx \\
    \overset{(a)}{\approx}& \int_{\rho_1}^{\rho_{K}} \frac{2(K-2) r^{1-\alpha}}{\rho_{K}^2 - \rho_{1}^2} dr = \frac{2(K-2)}{2-\alpha} \frac{\rho_{K}^{2-\alpha} - \rho_{1}^{2-\alpha}}{\rho_{K}^{2} - \rho_{1}^{2}}, \\
\end{split}
\label{eq-mean IK}
\end{equation}
where (a) simplifies the interference computation from $K-2$ PAs by considering that the nearest path between the target and the activated PA on the waveguide is the perpendicular distance $\rho_k$ instead of $d_{k,1}$. Since the location on the waveguide associated with $\rho_k$ follows a 2-D HPPP, the PDF of $\rho_k$ is: $f_{\rho_k}(r) = \frac{2r}{\rho_{K}^2 - \rho_{1}^2}$. By substituting $f_{\rho_k}(x)$ into \eqref{eq-mean IK}, the final expression is obtained. Based on a similar idea, the interference outside the circular $\bm{b}(\textit{O}, \rho_{K})$ is: $\mathbb{E} \{ \sum_{i=K+1}^{\infty}  d_{i,1}^{-\alpha} \} = 2\pi\lambda_l \int_{\rho_{K}}^{\infty} r^{-\alpha+1} dr = \frac{2\pi\lambda_{\rho}}{\alpha-2} \rho_{K}^{2-\alpha}$. For the case of $\alpha=2$, the result can be derived in a similar manner.

\vspace{-0.4cm}

\section{}
\label{Prop-RSS CRLB}

Let the denominator of the RSS-based CRLB in \eqref{eq-RSS CRLB exact} be $D = \bar{A}\bar{B} - \bar{C}^2$. We aim to select a distance $d_{*} = d_{p,q} \in \{ d_{k,m} \}_{k=1, m=0}^{K,M-1}$ that provides the highest amount of information to $D$ by maximizing the mutual information, which is: $I\left( D ; d_{k,m}|M,K\right) = h\left(D|M,K\right) - h\left(D|d_{k,m},M,K\right)$, whereas the differential entropies are given by
\begin{equation}
\begin{split}
     h\left(D|M,K\right) &= -\int_{\Bar{I}_D}f_{ D}(D|M,K) \\
    &\times\log_{2}{f_{D}( D|M,K)} {\rm{d}} D, \\
    \!\!\!\!h\left(D|d_{k,m},M,K\right) &= -\int_{\Bar{ I}_D}\int_{\Bar{\rm{D}}}  f_{D;d_{k,m}}(D,d_{k,m}|M,K) \\
    &\times\log_{2}{f_{D}(D|d_{k,m},M,K)} {\rm{d}} D {\rm{d}}d_{k,m}, \\ 
\end{split}
\label{eq-entropy}
\end{equation}
where $f_{D}(\cdot)$ and $f_{D;d_{k,m}}(\cdot)$ are the PDFs of $D$ and the joint PDF of $D$ and $d_{k,m}$, respectively, while $\Bar{I}_D$ and $\Bar{D}$ denote the supports of $D$ and $d_{k,m}$. These PDFs can be obtained through Monte-Carlo simulations. Based on the selected $q$-th PA on the waveguides, we have 
\begin{equation}
\begin{split}
\bar{A} &\approx M\sum_{k=1}^{K} \frac{(x_{k,q}-x_{t})^{2}}{d_{k,q}^{4}}, \
\bar{B} \approx M \sum_{k=1}^{K} \frac{(y_{k,q}-y_{t})^{2}}{d_{k,q}^{4}}, \\
%\bar{C} &\overset{(a)}{\approx} M^2 \left ( \sum_{i=1}^{K} \frac{(x_{i,q}-x_{t})(y_{i,q}-y_{t})}{d_{i,q}^{4}} \right )^2 \\
\bar{C} &\overset{(a)}{\approx} M^2 \sum_{k=1}^{K} \frac{(x_{k,q}-x_{t})^2(y_{k,q}-y_{t})^2}{d_{k,q}^{8}}, \\
\end{split}
\end{equation}
where (a) follows Sedrakyan’s inequality. By selecting the $p$-th waveguide based on mutual information, it yields
\begin{equation}
    \begin{split}
        b &\approx M \sum_{k=1}^{K} \frac{1}{d_{k,q}^2} \approx \frac{MK}{d_{p,q}^2}, \
        D 
        \overset{(b)}{\approx} M^2 K(K-1)/4d_{p,q}^{4}, \\
    \end{split}
\label{eq-b D}
\end{equation}
where (b) applies \cite[Proposition 1]{he-rss}. By substituting \eqref{eq-b D} into \eqref{eq-RSS CRLB exact}, the proof is completed.

\vspace{-0.4cm}

\bibliographystyle{ieeetr}
\bibliography{ref}

\end{document}